\documentclass[a4paper]{amsart}
\usepackage{amsmath,amsfonts,amssymb,amsthm,a4wide}

\theoremstyle{plain}
\newtheorem{lemma}{Lemma}
\newtheorem{theorem}{Theorem}
\newtheorem{corollary}[theorem]{Corollary}
\newtheorem{proposition}[theorem]{Proposition}
\theoremstyle{definition}
\newtheorem{definition}{Definition}
\newtheorem{remark}{Remark}
\newtheorem{observation}{Observation}

\newtheorem{example}{Example}

\title{Avoiding $2$-binomial squares and cubes}
\author[M. Rao]{Micha{\"e}l Rao}
\address[M. Rao]{CNRS, LIP, ENS Lyon, 46 allée d'Italie, 69364 Lyon Cedex 07 - France,
{\tt michael.rao@ens-lyon.fr}
}
\author[M. Rigo]{Michel Rigo}
\address[M. Rigo]{Dept of Math., University of Li{\`e}ge, Grande traverse 12 (B37), B-4000 Li{\`ege}, Belgium, {\tt M.Rigo@ulg.ac.be}}
\author[P. Salimov]{Pavel Salimov\dag{}}
\thanks{\dag{}The third author is supported by the Russian President's grant no. MK-4075.2012.1 and Russian Foundation for Basic Research grants no. 12-01-00089 and no. 11-01-00997 and by a University of Li{\`e}ge post-doctoral grant.}
\address[P. Salimov]{Dept of Math., University of Li{\`e}ge, Grande traverse 12 (B37), B-4000 Li{\`ege}, Belgium \and Sobolev Institute of Math., 4 Acad. Koptyug avenue, 630090 Novosibirsk, Russia}

\begin{document}

\begin{abstract}
    Two finite words $u,v$ are $2$-binomially equivalent if, for all
    words $x$ of length at most $2$, the number of occurrences of $x$ 
    as a (scattered) subword of $u$ is equal to the number of occurrences of $x$ in $v$. This notion
    is a refinement of the usual abelian equivalence. A $2$-binomial
    square is a word $uv$ where $u$ and $v$ are $2$-binomially
    equivalent.

    In this paper, considering pure morphic words, we prove that $2$-binomial squares (resp. cubes)
    are avoidable over a $3$-letter (resp. $2$-letter) alphabet. The sizes of the alphabets 
     are optimal.
\end{abstract}

\maketitle

\section{Introduction}

A {\em square} (resp. {\em cube}) is a non-empty word of the form $xx$ (resp. $xxx$). Since the work of Thue, it is well-known that there exists an infinite squarefree word over a ternary alphabet, and an infinite cubefree word over a binary alphabet \cite{Thue1,Thue2}. A main direction of research in combinatorics on words is about the avoidance of a pattern, and the size of the alphabet is a parameter of the problem. %Let $p$ be a finite word over a finite alphabet $V$. An infinite word $\mathbf{w}$ over a finite alphabet $A$ {\em avoids the pattern} $p$ if there does not exist a non-erasing morphism $h:V^*\to A^*$ such that $h(p)$ is a factor occurring in $\mathbf{w}$.

A possible and widely studied  generalization of squarefreeness is to consider an abelian framework. A non-empty word is an {\em abelian square} (resp. {\em abelian cube}) if it is of the form $xy$ (resp. $xyz$) where $y$ is a permutation of $x$ (resp. $y$ and $z$ are permutations of $x$). 
%In that context, the optimal size of the alphabet are given as follows. 
Erd\"os raised the question whether abelian squares can be avoided by an infinite word over an alphabet of size $4$~\cite{E}.
Ker{\"a}nen answered positively to this question, with a pure morphic word \cite{Ker}. % over a $4$-letter alphabet that avoids abelian squares 
Moreover Dekking has previously obtained an infinite word over a $3$-letter alphabet that avoids abelian cubes, and an infinite binary word that avoids abelian 4-powers \cite{Dek}.
(Note that in all these results, the size of the alphabet is optimal.)

In this paper, we are dealing with another generalization of squarefreeness and cubefreeness. We consider the $2$-binomial equivalence which is a refinement of the abelian equivalence, {\em i.e.}, if two words $x$ and $y$ are $2$-binomially equivalent, then $x$ is a permutation of $y$ (but in general, the converse does not hold, see Example~\ref{exa1} below). This equivalence relation is defined thanks to 
the binomial coefficient $\binom{u}{v}$ of two words $u$ and $v$ which is
the number of times $v$ occurs as a subsequence of $u$ (meaning as a
``scattered'' subword). For more on these binomial coefficients, see for instance \cite[Chap.~6]{Lot}. 
Based on this classical notion, the  $m$-binomial equivalence of two words has been recently introduced \cite{RS}.

\begin{definition}
  Let $m\in\mathbb{N}\cup\{+\infty\}$ and $u,v$ be two words over
  the alphabet $A$. We let $A^{\le m}$ denote the set of words of length at most $m$ over $A$. We say that $u$ and $v$ are {\em $m$-binomially equivalent} if 
  $$\binom{u}{x}=\binom{v}{x},\ \forall x\in A^{\le m}.$$
  We simply write $u\sim_m v$ if $u$ and $v$ are $m$-binomially equivalent. The word $u$ is obtained as a permutation of the letters in $v$ if and only if $u\sim_1 v$. In that case, we say that $u$ and $v$ are {\em abelian equivalent} and we write instead $u\sim_{\mathsf{ab}}v$. Note that if $u\sim_{k+1}v$, then $u\sim_k v$, for all $k\ge 1$.
\end{definition}

\begin{example}\label{exa1}
    The four words $0101110$, $0110101$, $1001101$ and
    $1010011$ are $2$-binomially equivalent. Let $u$ be any of these four words. We have
$$\binom{u}{0}=3,\ \binom{u}{1}=4,\ \binom{u}{00}=3,\ \binom{u}{01}=7,\ \binom{u}{10}=5,\ \binom{u}{11}=6.$$
For instance, the word $0001111$ is abelian equivalent to $0101110$ but these two words are not $2$-binomially equivalent. Let $a$ be a letter. It is clear that $\binom{u}{aa}$ and $\binom{u}{a}$ carry the same information, {\em i.e.}, $\binom{u}{aa}=\binom{|u|_a}{2}$ where $|u|_a$ is the number of occurrences of $a$ in $u$.
\end{example}

A {\em $2$-binomial square} (resp. {\em $2$-binomial cube}) is a non-empty word of the form $xy$ where $x\sim_2 y$ (resp. $x\sim_2 y\sim_2 z$). Squares are avoidable over a $3$-letter alphabet and abelian squares are avoidable over a $4$-letter
alphabet. Since $2$-binomial equivalence lies between abelian equivalence and equality, the question is to determine whether or not $2$-binomial squares are avoidable over a $3$-letter alphabet. We answer positively to this question in Section~\ref{sec2}. The fixed point of the morphism $g:0\mapsto 012, 1\mapsto 02, 2\mapsto 1$ avoids $2$-binomial squares.

In a similar way, cubes are avoidable over a $2$-letter alphabet and abelian squares are avoidable over a $3$-letter
alphabet. The question is to determine whether or not $2$-binomial cubes are avoidable over a $2$-letter alphabet. We also answer positively to this question in Section~\ref{sec3}. The fixed point of the morphism $h:0\mapsto 001, 1\mapsto 011$ avoids $2$-binomial cubes.

\begin{remark}
The $m$-binomial equivalence is not the only way to refine the abelian equivalence. Recently, a notion of $m$-abelian equivalence has been introduced \cite{KSZ}. To define this equivalence, one counts the number $|u|_x$ of occurrences in $u$ of all factors $x$ of length up to $m$ (it is meant factors made of consecutive letters). 
That is $u$ and $v$ are $m$-abelian equivalent if $|u|_x=|v|_x$ for all $x\in A^{\le m}$. 
In that context, the results on avoidance are quite different. 
Over a $3$-letter alphabet $2$-abelian squares are unavoidable:  the longest ternary word which is $2$-abelian squarefree has length $537$ \cite{KH}, and pure morphic words cannot avoid $k$-abelian-squares for every $k$ \cite{9}. 
On the other hand, it has been shown that there exists  a $3$-abelian squarefree morphic word over a $3$-letter alphabet \cite{R}.
Moreover $2$-abelian-cubes can be avoided over a binary alphabet by a morphic word \cite{R}. % while $1$-abelian-cubes which are the abelian-cubes, cannot be avoided.
%It has been shown that there exists, over a $3$-letter alphabet, a $64$-abelian squarefree morphic word \cite{Huo}. But it is conjectured that there exists, over $3$-letter alphabet, a $3$-abelian squarefree word \cite{9}. 
\end{remark}

The number of occurrences of a letter $a$ in a word $u$ will be denoted either by $\binom{u}{a}$ or $|u|_a$. Let $A=\{0,1,\ldots,k\}$ be an alphabet. The {\em Parikh map} is an application $\Psi:A^*\to\mathbb{N}^{k+1}$ such that $\Psi(u)=(|u|_0,\ldots,|u|_{k})^T$. Note that we will deal with column vectors (when multiplying a square matrix with a column vector on its right). In particular, two words are abelian equivalent if and only if they have the same Parikh vector. The mirror of the word $u=u_1u_2\cdots u_k$ is denoted by $\widetilde{u}=u_k\cdots u_2u_1$.

\section{Avoiding $2$-binomial squares over a $3$-letter alphabet}\label{sec2}
Let $A=\{0,1,2\}$ be a $3$-letter alphabet. Let $g:A^*\to A^*$ be the morphism defined by 
$$g:\left\{
\begin{array}{lll}
0&\mapsto&012\\
1&\mapsto&02\\
2&\mapsto&1\\
\end{array}\right. \text{ and thus, }
g^2:\left\{
\begin{array}{lll}
0&\mapsto&012021\\
1&\mapsto&0121\\
2&\mapsto&02.\\
\end{array}\right.$$
It is prolongable on $0$: $g(0)$ has $0$ as a prefix. Hence the limit $\mathbf{x}=\lim_{n\to +\infty} g^n(0)$ is a well-defined infinite word 
$$\mathbf{x}=g^\omega(0)=012021012102012021020121\cdots$$
which is a fixed point of $g$. Since the original work of Thue, this word $\mathbf{x}$ is well-known to avoid (usual) squares. It is sometimes referred to as the {\em ternary Thue--Morse word}. We will make use of the fact that $X=\{012,02,1\}$ is a prefix-code and thus an $\omega$-code: Any finite word in $X^*$ (resp. infinite word in $X^\omega$) has a unique factorization as a product of elements in $X$. Let us make an obvious but useful observation.

\begin{observation}\label{obs:02}
The factorization of $\mathbf{x}$ in terms of the elements in $X$ permits to write $\mathbf{x}$ as  
$$\mathbf{x}=0\, \alpha_1\, 2\, \alpha_2\, 0\, \alpha_3\, 2\, \alpha_4\, 0\, \alpha_5\, 2\, \alpha_6\, 0\cdots$$
where, for all $i\ge 1$, $\alpha_i\in\{\varepsilon,1\}$. 
That is, the image of $\mathbf{x}$ by the morphism $e: 0\mapsto 0, 1\mapsto \varepsilon, 2 \mapsto 2$ (which erases all the $1$'s) is $e(\mathbf{x})=(02)^{\omega}$.
%Roughly speaking, if we erase all the $1$'s occurring in $\mathbf{x}$, we get the word $(02)^\omega$. 
%The same observation holds for $\mathbf{y}$.
\end{observation}

The next property is well known. 
For example, it comes from the fact that the image of the ternary Thue--Morse word by the morphism $0\mapsto 011, 1 \mapsto 01, 2\mapsto 0$ is the Thue--Morse word.
However, for the sake of completeness, we give a direct proof here.

\begin{lemma}\label{lem:mirror}
  A word $u$ is a factor occurring in $\mathbf{x}$ if and only if $\widetilde{u}$ is a factor occurring in $\mathbf{x}$. 
\end{lemma}
\begin{proof}
We define the morphism $\widetilde{g}:A^*\to A^*$ by considering the mirror images of the images of the letters by $g$,
$$\widetilde{g}:\left\{
\begin{array}{lll}
0&\mapsto&210\\
1&\mapsto&20\\
2&\mapsto&1\\
\end{array}\right.
\text{ and thus, }
\widetilde{g}^2:\left\{
\begin{array}{lll}
0&\mapsto&120210\\
1&\mapsto&1210\\
2&\mapsto&20.\\
\end{array}\right.
$$
Note that $\widetilde{g}$ is not prolongable on any letter. But 
the morphism $\widetilde{g}^2$ is prolongable on the letter $1$. We consider the infinite word
$$\mathbf{y}=(\widetilde{g}^2)^\omega(1)=1210 20 1210 120210 20 120210
1210\cdots.$$
%
%Let $a\in A$ be a letter. 
If $v\in A^*$ is a non-empty word ending with $a\in A$, {\em i.e.}, $v=ua$ for some word $u\in A^*$, we denote by $va^{-1}$ the word obtained by removing the suffix $a$ from $v$. So $va^{-1}=u$. 
%Similarly, if $v=au$, then $a^{-1}v$ is equal to $u$. 

For every words $r$ and $s$ %e words. %Let us make an obvious but useful observation. 
we have $r=g^2(s)\Leftrightarrow \widetilde{r}=\widetilde{g}^2(\widetilde{s})$. 
%In general, for all $n\ge 1$, we have $r=g^n(s)\Leftrightarrow \widetilde{r}=\widetilde{g}^n(\widetilde{s})$.
Obviously, $u$ is a factor occurring in $\mathbf{x}$ if and only if $\widetilde{u}$ is a factor occurring in $\mathbf{y}$.  

On the other hand, $\widetilde{g}^2$ is a cyclic shift of $g^2$, since $g^2(a)= 0 \widetilde{g}^2(a) 0^{-1}$ for every $a\in \{0,1,2\}$. Thus $u$ is a factor occurring in $\mathbf{x}$ if and only if $u$ is a factor occurring in $\mathbf{y}$. 
To summarize, $u$ is a factor occurring in $\mathbf{x}$ if and only if $u$ is a factor occurring in $\mathbf{y}$, and $u$ is a factor occurring in $\mathbf{y}$ if and only if $\widetilde{u}$ is a factor occurring in $\mathbf{x}$. This concludes the proof.
\end{proof}

%% Note that $\widetilde{g}^2(i)$ ends with $0$ for $i=0,1,2$. Observe that the following conjugacy relations holds, images by $g^2$ are cyclic shifts of the corresponding images by $\widetilde{g}^2$,
%% \begin{equation}
%%     \label{eq:rel}
%%     0\, \widetilde{g}^2(i)\, 0^{-1}= g^2(i),\ \forall i\in A.
%% \end{equation}
%% % $$\begin{array}{rcl}
%% % 0\, \widetilde{g}^2(0)\, 0^{-1}&=&012|02|1= g^2(0)\\
%% % 0\, \widetilde{g}^2(1)\, 0^{-1}&=&012|1 = g^2(1)\\
%% % 0\, \widetilde{g}^2(2)\, 0^{-1}&=&02=g^2(2).\\
%% % \end{array}$$
%% Since $\mathbf{x}$ is a fixed point of $g$, $\mathbf{x}$ belongs to $X^\omega$. Moreover, from the above relations, we get 
%% $$    \mathbf{y}=0\, \widetilde{g}^2(1)\, \widetilde{g}^2(2)\, \widetilde{g}^2(1)\, \widetilde{g}^2(0)\cdots=g^2(1)\, g^2(2)\, g^2(1)\, g^2(0)\cdots$$
%% and thus, the word $\mathbf{y}$ also belongs to $X^\omega$.

We will be dealing with $2$-binomial squares so, in particular, with abelian squares. The next lemma %is the key statement of this section. It will 
permit to ``desubstitute'', meaning that we are looking for the inverse image of a factor under the considered morphism.

\begin{lemma}\label{lem2}
  Let $u,v\in A^*$ be two abelian equivalent non-empty words such that $uv$ is a factor occurring in $\mathbf{x}$. There exists $u',v'\in A^*$ such that $u'v'$ is a factor of $\mathbf{x}$, and either:
  \begin{enumerate}
    \item $u=g(u')$ and $v=g(v')$;
    \item or, $\widetilde{u}=g(v')$ and $\widetilde{v}=g(u')$.
  \end{enumerate}
\end{lemma}

\begin{proof}
We will make an extensive use of Observation~\ref{obs:02}. 
Note that $u$ and $v$ must contain at least one $0$ or one $2$.
Obviously $e(uv)$ is an abelian square of $(02)^{\omega}$, thus either $e(u)=e(v)=(02)^i$ or $e(u)=e(v)=(20)^i$ for an $i>0$.

If $e(u)=e(v)=(02)^i$, then we have $u=a \, 0\, \cdots \, 2\, b$  and $v=c \, 0\, \cdots \, 2\, d$ with $a,bc,d\in\{\varepsilon,1\}$. In this case, we deduce that $u$ and $v$ belongs to $X^*$. Otherwise stated, since $uv$ is a factor of $\mathbf{x}$,  there exists a factor $u'v'$ in $\mathbf{x}$  such that $g(u')=u$ and $g(v')=v$.

Otherwise we have $e(u)=e(v)=(20)^i$. Thanks to Lemma~\ref{lem:mirror}, $\widetilde{v}\widetilde{u}$ is a factor occurring in $\mathbf{x}$, and $e(\widetilde{u})=e(\widetilde{v})=(02)^i$. Thus we are reduced to the previous case, and there is a factor $u',v'$ in $\mathbf{x}$  such that $g(u')=\widetilde{v}$ and $g(v')=\widetilde{u}$.
\end{proof}

Let $u$ be a word. We set
$$\lambda_u:=\binom{u}{01}-\binom{u}{12}.$$ When we use the desubstitution provided by the previous lemma, the shorter factors $u'$ and $v'$ derived from $u$ and $v$ keep properties from their ancestors.

\begin{lemma}
Let $u,v\in A^*$ be two abelian equivalent non-empty words such that $uv$ is a factor occurring in $\mathbf{x}$. Let $u',v'$ be given by Lemma~\ref{lem2}.
If $\lambda_u=\lambda_v$, then $u'$ and $v'$ are abelian equivalent and $\lambda_{u'}=\lambda_{v'}$.
\end{lemma}

\begin{proof}
If we are in the second situation described by Lemma~\ref{lem2}, then $\widetilde{v}\widetilde{u}$ is also a factor occurring in $\mathbf{x}$. Obviously $\widetilde{v}$ and $\widetilde{u}$ are also abelian equivalent, $\lambda_{\widetilde{v}}=\lambda_{\widetilde{u}}$ and the case is reduced to the first situation.

Assume now w.l.o.g. that we are in the first situation, % described by Lemma~\ref{lem2}:  $u'v'$ occurs in $\mathbf{x}$, $u=g(u')$ and $v=g(v')$.
that is $u=g(u')$ and $v=g(v')$.
First observe that we have, for all $a,b\in A$, $a\neq b$,
\begin{equation}
    \label{eq:delta}
\binom{u'}{ab}=\binom{|u'|_a+|u'|_b}{2}-\binom{|u'|_a}{2}-\binom{|u'|_b}{2}-\binom{u'}{ba}.    
\end{equation}
Since $u=g(u')$, we derive that
$$\binom{u}{01}=|u'|_0+\binom{u'}{00}+\binom{u'}{02}+\binom{u'}{12}+\binom{|u'|_0+|u'|_1}{2}-\binom{|u'|_0}{2}-\binom{|u'|_1}{2}-\binom{u'}{01},
$$
\begin{eqnarray*}
\binom{u}{12}&=&|u'|_0+\binom{u'}{00}+\binom{u'}{01}+\binom{|u'|_1+|u'|_2}{2}-\binom{|u'|_1}{2}-\binom{|u'|_2}{2}-\binom{u'}{12}\\
&&+\binom{|u'|_0+|u'|_2}{2}-\binom{|u'|_0}{2}-\binom{|u'|_2}{2}-\binom{u'}{02}.
\end{eqnarray*}
Hence
$$\lambda_{u}=
2\biggl[\binom{u'}{02}-\binom{u'}{01}+\binom{u'}{12}-\binom{|u'|_2}{2}\biggr]+\binom{|u'|_0+|u'|_1}{2}-\binom{|u'|_1+|u'|_2}{2}-\binom{|u'|_0+|u'|_2}{2}.
$$
Similar relations holds for $v$.

Since $u'$ and $v'$ occur in $\mathbf{x}$, from Observation~\ref{obs:02}, we get
\begin{equation}
    \label{eq:star}
    \left| |u'|_0-|u'|_2 \right|\le 1 \text{ and }
\left| |v'|_0-|v'|_2 \right|\le 1.
\end{equation}
Since $u\sim_{\mathsf{ab}}v$, we have $|u|_1=|v|_1$. Hence, from the definition of $g$, $|u'|_0+|u'|_2=|v'|_0+|v'|_2$. In the same way, $|u|_2=|v|_2$ implies that $|u'|_0+|u'|_1=|v'|_0+|v'|_1$ or equivalently, 
$|u'|_1-|v'|_1=|v'|_0-|u'|_0$. From the above relation and \eqref{eq:star}, we get
$$\left| |v'|_0-|u'|_0 + |u'|_2-|v'|_2 \right|\le 2 \text{ and }
|u'|_2-|v'|_2=|v'|_0-|u'|_0.$$
Hence the difference of the following two Parikh vectors can only take three values
$$\Psi(u')-\Psi(v')\in\left\{
\begin{pmatrix}
    0\\0\\0\\
\end{pmatrix},  
\begin{pmatrix}
    1\\-1\\-1\\
\end{pmatrix}, 
\begin{pmatrix}
    -1\\1\\1\\
\end{pmatrix}\right\}.$$
To prove that $u'$ and $v'$ are abelian equivalent, we will rule out the last two possibilities.

By assumption, $\lambda_u=\lambda_v$. So this relation also holds modulo $2$. Hence
\begin{eqnarray*}
&&\binom{|u'|_0+|u'|_1}{2}-\binom{|u'|_1+|u'|_2}{2}-\binom{|u'|_0+|u'|_2}{2}\\
&\equiv& \binom{|v'|_0+|v'|_1}{2}-\binom{|v'|_1+|v'|_2}{2}-\binom{|v'|_0+|v'|_2}{2} \pmod{2}.
\end{eqnarray*}
Assume that we have 
$$\Psi(u')-\Psi(v')=
\begin{pmatrix}
    1\\-1\\-1\\
\end{pmatrix}, i.e., 
\begin{array}{rcl}
|u'|_0+|u'|_1&=&|v'|_0+|v'|_1,\\
|u'|_0+|u'|_2&=&|v'|_0+|v'|_2,\\
|u'|_1+|u'|_2&=&|v'|_1+|v'|_2-2.\\
\end{array}$$
This leads to a contradiction because then
$$\binom{|u'|_1+|u'|_2}{2}\not\equiv \binom{|v'|_1+|v'|_2}{2} \pmod{2}.$$
Indeed, it is easily seen that $\binom{4n}{2}\equiv 0\pmod{2}$, $\binom{4n+1}{2}\equiv 0\pmod{2}$, $\binom{4n+2}{2}\equiv 1\pmod{2}$ and $\binom{4n+3}{2}\equiv 1\pmod{2}$.

The case $\Psi(u')-\Psi(v')=
\left ( \begin{smallmatrix}
    -1\\1\\1\\
\end{smallmatrix}\right )$ is handled similarly.
So we can assume now that $\Psi(u')=\Psi(v')$, that is $u'\sim_{\mathsf{ab}} v'$. 
It remains to prove that $\lambda_{u'}=\lambda_{v'}$. 
By assumption $\lambda_u=\lambda_v$, % and we have just shown that $u'\sim_{\mathsf{ab}} v'$. 
and from the above formula describing $\lambda_u$ (resp. $\lambda_v$) we get
$$\binom{u'}{02}-\binom{u'}{01}+\binom{u'}{12}=\binom{v'}{02}-\binom{v'}{01}+\binom{v'}{12}.$$
To conclude that $\lambda_{u'}=\lambda_{v'}$, we should simply show that $\binom{u'}{02}=\binom{v'}{02}$. But $u'v'$ is a factor occurring in $\mathbf{x}$ (from Observation~\ref{obs:02}, when discarding the $1$'s with just alternate $0$'s and $2$'s) and $u'\sim_{\mathsf{ab}}v'$. This concludes the proof.
%
%% Assume now that we are in the second situation described by Lemma~\ref{lem2}:  $v'u'$ occurs in $\mathbf{x}$, $\widetilde{u}=g(u')$ and $\widetilde{v}=g(v')$. Observe that
%% $$\lambda_{\widetilde{u}}=\binom{\widetilde{u}}{01}-\binom{\widetilde{u}}{12}=
%% \binom{u}{10}-\binom{u}{21}.$$
%% Using \eqref{eq:delta} twice, this quantity can be written as
%% $$\lambda_{\widetilde{u}}=\binom{|u|_0+|u|_1}{2}-\binom{|u|_1+|u|_2}{2}-\binom{|u|_0}{2}+\binom{|u|_2}{2}-\lambda_u.$$
%% Since $\lambda_u=\lambda_v$ and $u\sim_{\mathsf{ab}}v$, we get $\lambda_{\widetilde{u}}=\lambda_{\widetilde{v}}$. We can follow exactly the same proof as in the first situation.
\end{proof}

\begin{theorem}
    The word $\mathbf{x}=g^\omega(0)=012021012102012021020121\cdots$ avoids $2$-binomial squares.
\end{theorem}

\begin{proof}
Assume to the contrary that $\mathbf{x}$ contains a $2$-binomial square $uv$ where $u$ and $v$ are $2$-binomially equivalent. In particular, $u$ and $v$ are abelian equivalent and moreover $\lambda_u=\lambda_v$. We can therefore apply iteratively Lemma~\ref{lem2} and the above lemma to words of decreasing lengths and get finally a repetition $aa$ with $a\in A$ in $\mathbf{x}$. But $\mathbf{x}$ does not contain any such factor.
\end{proof}

\begin{remark}
The fixed point of $g$ is $2$-binomial-square free, but $g$ is not $2$-binomial-square-free, that is the image of a $2$-binomial-square-free word may contain a $2$-binomial-square (\emph{e.g.}, $g(010)=01202012$ contains the square $2020$).
\end{remark}

\section{Avoiding $2$-binomial cubes over a $2$-letter alphabet}\label{sec3}

Consider the morphism $h:0\mapsto 001$ and $h:1\mapsto 011$. 
In this section, we show that $h$ is $2$-binomial-cube-free, that is for every $2$-binomial-cube free binary word $w$, $h(w)$ is $2$-binomial-cube-free.
%We claim that the infinite word 
As a direct corollary, we get that the fixed point of $h$, 
$$\mathbf{z}=h^\omega(0)%=z_0z_1z_2\cdots
=001 001011 001001011001011011\cdots$$
avoids $2$-binomial cubes. 
%Note that $z_{3k}=0$ and $z_{3k+2}=1$ for all $k\ge 0$. 
%We will make use of the following observations.

Let $u$ be a word over $\{0,1\}$. The {\em extended Parikh vector} of $u$ is 
$$\Psi_2(u)=\biggl(|u|_0,|u|_1,\binom{u}{00},\binom{u}{01},\binom{u}{10},\binom{u}{11}\biggr)^T.$$
Observe that two words $u$ and $v$ are $2$-binomially equivalent if and only if $\Psi_2(u)=\Psi_2(v)$.

Consider the matrix $M_h$ given by
$$M_h=
\begin{pmatrix}
2&1&0&0&0&0\\
1&2&0&0&0&0\\
1&0&4&2&2&1\\
2&2&2&4&1&2\\
0&0&2&1&4&2\\
0&1&1&2&2&4\\
\end{pmatrix}.$$
One can check that $M_h$ is invertible. 
%The matrix $M_h$ has the following property:
We will make use of the following observations:
\begin{proposition}\label{prop:mh}
For every $u\in\{0,1\}^*$,
$$\Psi_2(h(u))=M_h\Psi_2(u).$$ %\ \forall u\in\{0,1\}^*.$$
\end{proposition}

\begin{proposition}\label{prop:cyc}
Let $u=1x$ and $u'=x1$ be two words over $\{0,1\}$. We have $|u|_0=|u'|_0$, $|u|_1=|u'|_1$, 
$$\binom{u}{00}= \binom{u'}{00},\ \binom{u}{11}= \binom{u'}{11},\  \binom{u'}{01}= \binom{u}{01}+|u|_0,\   \binom{u'}{10}= \binom{u}{10}-|u|_0.$$  
In particular, if $1x\sim_2 1y$, then $x1\sim_2 y1$. Similar relations hold for $0x$ and $x0$. In particular, if $x0\sim_2 y0$, then $0x\sim_2 0y$.
\end{proposition}

Let $x,y\in\{0,1\}$. We set $\delta_{x,y}=1$, if $x=y$; and $\delta_{x,y}=0$, otherwise.

\begin{lemma}\label{lem:cubeone}
Let $p'$, $q'$ and $r'$ be binary words, and let $a,b\in \{0,1\}$. 
Let $p=h(p')\,0$, $q= a\, 1\, h(q')\, 0\, b$ and $r=1\, h(r')$.
Then either $p \not \sim_2 q$ or $p \not\sim_2 r$.
\end{lemma}
\begin{proof}
Assume, for the sake of contradiction, that $p\sim_2 q\sim_2 r$. %$h(p')\,0 \sim_2 a\, 1\, h(q')\, 0\, b \sim_2 1\, h(r')$.
Then $\vert p'\vert =\vert q'\vert +1 = \vert r'\vert = n$.
The following relations can mostly be derived from the coefficients of $M_h$ (we also have to take into account the extra suffix $0$ of $p$, respectively the extra prefix $1$ in $r$): 
$$\binom{p}{01}=2\binom{p'}{0}+2\binom{p'}{1}+2\binom{p'}{00}+4\binom{p'}{01}+\binom{p'}{10}+2\binom{p'}{11},$$
$$\binom{p}{10}=\binom{p'}{0}+2\binom{p'}{1}+2\binom{p'}{00}+\binom{p'}{01}+4\binom{p'}{10}+2\binom{p'}{11},$$
$$\Rightarrow \binom{p}{01}-\binom{p}{10}=\binom{p'}{0}+3\binom{p'}{01}-3\binom{p'}{10};$$
$$\binom{r}{01}=2\binom{r'}{0}+2\binom{r'}{1}+2\binom{r'}{00}+4\binom{r'}{01}+\binom{r'}{10}+2\binom{r'}{11},$$
$$\binom{r}{10}=2\binom{r'}{0}+\binom{r'}{1}+2\binom{r'}{00}+\binom{r'}{01}+4\binom{r'}{10}+2\binom{r'}{11},$$
$$\Rightarrow \binom{r}{01}-\binom{r}{10}=\binom{r'}{1}+3\binom{r'}{01}-3\binom{r'}{10}.$$

We also get the following relations:
\begin{eqnarray*}
\binom{q}{01}&=&2\binom{q'}{0}+2\binom{q'}{1}+2\binom{q'}{00}+4\binom{q'}{01}+\binom{q'}{10}+2\binom{q'}{11}\\
&&+\delta_{a,0}\biggl[ 1+\binom{q'}{0}+2\binom{q'}{1}+\delta_{b,1}\biggr]+\delta_{b,1}\biggl[ 1+2\binom{q'}{0}+\binom{q'}{1}\biggr],
\end{eqnarray*}
\begin{eqnarray*}
\binom{q}{10}&=&3\binom{q'}{0}+3\binom{q'}{1}+2\binom{q'}{00}+\binom{q'}{01}+4\binom{q'}{10}+2\binom{q'}{11}+1\\
&&+\delta_{a,1}\biggl[1+\delta_{b,0}+2\binom{q'}{0}+\binom{q'}{1}\biggr]+
\delta_{b,0}\biggl[1+\binom{q'}{0}+2\binom{q'}{1}\biggr]\\
&=&(6-2\delta_{a,0}-\delta_{b,1})\binom{q'}{0}+(6-\delta_{a,0}-2\delta_{b,1})\binom{q'}{1}+4-2\delta_{a,0}-2\delta_{b,1}+\delta_{a,0}\delta_{b,1}\\
&&+2\binom{q'}{00}+\binom{q'}{01}+4\binom{q'}{10}+2\binom{q'}{11}.
\end{eqnarray*}
Where for the last equality, we have used the fact that $\delta_{a,1}=1-\delta_{a,0}$ and $\delta_{b,0}=1-\delta_{b,1}$.
Finally, we obtain
$$\binom{q}{01}-\binom{q}{10}=
(-4+3\delta_{a,0}+3\delta_{b,1}) \biggl[ \binom{q'}{0}+\binom{q'}{1}\biggr] +3\binom{q'}{01}-3\binom{q'}{10}
-4+3\delta_{a,0}+3\delta_{b,1}.
$$

Since $p\sim_2 q\sim_2 r$, we have $\binom{p}{10}-\binom{p}{01}=\binom{q}{10}-\binom{q}{01}=\binom{r}{10}-\binom{r}{01}$. In particular, these equalities modulo $3$ give
\begin{equation}
    \label{eq:mod3}
\binom{p'}{0}\equiv \binom{r'}{1}\equiv 2\biggr[\binom{q'}{0}+\binom{q'}{1}+1\biggl]\equiv 2n \pmod{3}.    
\end{equation}
Now, we take into account the fact that $p$ and $r$ are abelian equivalent to get a contradiction. Since $p=h(p')\, 0$ and $r=1\, h(r')$, we get
$$
\begin{pmatrix}
    |p|_0\\ |p|_1\\
\end{pmatrix}=
\begin{pmatrix}
    2&1\\ 1&2\\
\end{pmatrix}
\begin{pmatrix}
    |p'|_0\\ |p'|_1\\
\end{pmatrix}+
\begin{pmatrix}
    1\\0\\
\end{pmatrix},\ 
\begin{pmatrix}
    |r|_0\\ |r|_1\\
\end{pmatrix}=
\begin{pmatrix}
    2&1\\ 1&2\\
\end{pmatrix}
\begin{pmatrix}
    |r'|_0\\ |r'|_1\\
\end{pmatrix}+
\begin{pmatrix}
    0\\1\\
\end{pmatrix}.$$
Hence, we obtain
$$
\begin{pmatrix}
    |p|_0- |r|_0\\ |p|_1-|r|_1\\
\end{pmatrix}=
\begin{pmatrix}
    0\\0\\
\end{pmatrix}
=\begin{pmatrix}
    2&1\\ 1&2\\
\end{pmatrix}\begin{pmatrix}
    |p'|_0- |r'|_0\\ |p'|_1-|r'|_1\\
\end{pmatrix}+\begin{pmatrix} 1\\-1\\
\end{pmatrix}.
$$
We derive that $|p'|_0-|r'|_0=-1$ and $|p'|_1-|r'|_1=1$. Recalling that $|p'|_0+|p'|_1=n$. If we subtract the last two equalities, we get $|p'|_0+|r'|_1=n-1$. From \eqref{eq:mod3}, we know that $|p'|_0\equiv |r'|_1\pmod{3}$. Hence $2|p'|_0\equiv n-1\pmod{3}$ and thus
$$|p'|_0\equiv 2n-2\pmod{3}.$$
This contradicts the fact again given by \eqref{eq:mod3} that $|p'|_0\equiv 2n \pmod{3}$.
\end{proof}

Similarly, one get the following lemma.

\begin{lemma}\label{lem:cubetwo}
Let $p'$, $q'$ and $r'$ be binary words, and let $a,b\in \{0,1\}$. 
Let $p=h(p')\,0\, a$, $q=1\, h(q')\, 0$ and $r=b\, 1\, h(r')$.
Then either $p \not \sim_2 q$ or $p \not\sim_2 r$.
\end{lemma}
\begin{proof}
Assume, for the sake of contradiction, that $p\sim_2 q\sim_2 r$. %$h(p')\,0 \sim_2 a\, 1\, h(q')\, 0\, b \sim_2 1\, h(r')$.
Then $\vert p'\vert =\vert q'\vert = \vert r'\vert = n$.
%$$p=h(p')\,0\, a,\ 
%q=1\, h(q')\, 0,\ 
%r=b\, 1\, h(r'),\quad a,b\in\{0,1\}, |p'|=|q'|=|r'|=n.$$
Taking into account the special form of $p$ and $q$, we get
$$\binom{p}{01}=2\binom{p'}{0}+2\binom{p'}{1} +2\binom{p'}{00}+4\binom{p'}{01}+\binom{p'}{10}+2\binom{p'}{11}+\delta_{a,1}
\biggr(1+2\binom{p'}{0}+\binom{p'}{1}\biggr),$$
$$\binom{p}{10}=\binom{p'}{0}+2\binom{p'}{1}+2\binom{p'}{00}+\binom{p'}{01}+4\binom{p'}{10}+2\binom{p'}{11}+\delta_{a,0}
\biggr(\binom{p'}{0}+2\binom{p'}{1}\biggr),$$
$$\binom{q}{01}=2\binom{q'}{0}+2\binom{q'}{1} +2\binom{q'}{00}+4\binom{q'}{01}+\binom{q'}{10}+2\binom{q'}{11},$$
$$\binom{q}{10}=3\binom{q'}{0}+3\binom{q'}{1}+2\binom{q'}{00}+\binom{q'}{01}+4\binom{q'}{10}+2\binom{q'}{11}+1.$$
Hence, we get
$$\binom{p}{01}-\binom{p}{10}=-2\binom{p'}{1} +3\binom{p'}{01}-3\binom{p'}{10}+\delta_{a,1}
\biggr(1+3\binom{p'}{0}+3\binom{p'}{1}\biggr),$$
$$\binom{q}{01}-\binom{q}{10}=-\binom{q'}{0}-\binom{q'}{1}+3\binom{q'}{01}-3\binom{q'}{10}-1.$$
Since, $p\sim_2 q$, the last two relations evaluated modulo $3$ give
\begin{equation}
    \label{eq:t1}
    |p'|_1+\delta_{a,1}\equiv 2n+2 \pmod{3}.
\end{equation}
Similarly, the form of $r$ gives the following relations
$$\binom{r}{01}=2\binom{r'}{0}+2\binom{r'}{1} +2\binom{r'}{00}+4\binom{r'}{01}+\binom{r'}{10}+2\binom{r'}{11}+\delta_{b,0}
\biggr(1+\binom{r'}{0}+2\binom{r'}{1}\biggr),$$
$$\binom{r}{10}=2\binom{r'}{0}+\binom{r'}{1}+2\binom{r'}{00}+\binom{r'}{01}+4\binom{r'}{10}+2\binom{r'}{11}
+\delta_{b,1} \biggr(2\binom{r'}{0}+\binom{r'}{1}\biggr),$$
$$\binom{r}{01}-\binom{r}{10}=-2\binom{r'}{0}+3\binom{r'}{01}-3\binom{r'}{10}+\delta_{b,0}
\biggr(1+3\binom{r'}{0}+3\binom{r'}{1}\biggr)$$
Since, $p\sim_2 r$, the last two relations evaluated modulo $3$ give
\begin{equation}
    \label{eq:t2}
    |p'|_1+\delta_{a,1}\equiv |r'|_0+\delta_{b,0} \pmod{3}.
\end{equation}
Now, we take into account the fact that $p$, $q$ and $r$ are abelian equivalent to get a contradiction. The following two vectors are equal:
$$
\begin{pmatrix}
    |p|_0\\ |p|_1\\
\end{pmatrix}=
\begin{pmatrix}
    2&1\\ 1&2\\
\end{pmatrix}
\begin{pmatrix}
    |p'|_0\\ |p'|_1\\
\end{pmatrix}+
\begin{pmatrix}
    1+\delta_{a,0}\\ \delta_{a,1}\\
\end{pmatrix},\ 
\begin{pmatrix}
    |r|_0\\ |r|_1\\
\end{pmatrix}=
\begin{pmatrix}
    2&1\\ 1&2\\
\end{pmatrix}
\begin{pmatrix}
    |r'|_0\\ |r'|_1\\
\end{pmatrix}+
\begin{pmatrix}
    \delta_{b,0}\\1+\delta_{b,1}\\
\end{pmatrix}.$$
We derive easily that
$$
 |p'|_1-|r'|_1=\ 1+\delta_{a,0}-\delta_{b,0}.
$$
On the one hand, using the latter relation and \eqref{eq:t2}
$$|r'|_1+1+\delta_{a,0}-\delta_{b,0}+\delta_{a,1}=|p'|_1+\delta_{a,1}\equiv |r'|_0+\delta_{b,0} \pmod{3}$$
Replacing $|r'|_0$ by $n-|r'|_1$, we get $2|r'|_1+2\equiv n+2\delta_{b,0} \pmod{3}$, or equivalently
$$|r'|_1+1\equiv 2n+\delta_{b,0} \pmod{3}.$$
On the other hand, using \eqref{eq:t1},
$$|r'|_1+1+\delta_{a,0}-\delta_{b,0}+\delta_{a,1}=|p'|_1+\delta_{a,1}\equiv 2n+2\pmod{3}$$
and thus,
$$|r'|_1\equiv 2n+\delta_{b,0}\pmod{3}.$$
We get a contradiction, $2n+\delta_{b,0}$ should congruent to both $|r'|_1$ and $|r'|_1+1$ modulo $3$.
\end{proof}

We are ready to prove the main theorem of this section.

\begin{theorem}
%    The infinite word $\mathbf{z}=001001011\cdots$ fixed point of $h:0\mapsto 001,1\mapsto 011$ avoids $2$-binomial cubes.
Let $h:0\mapsto 001,1\mapsto 011$. 
For every $2$-binomial-cube-free word $w\in\{0,1\}^*$, $h(w)$ is $2$-binomial-cube-free.
\end{theorem}

\def\wt{h(w)}
\begin{proof}
%$\mathbf{z}$
Let $w$ be a $2$-binomial-cube-free binary word.
Assume that $h(w)= z_0 \ldots z_{3\vert
w\vert -1}$ contains a $2$-binomial cube $pqr$ occurring in position $i$, {\em i.e.}, $p\sim_2 q\sim_2 r$ and $w=w'\, p\, q\, r\, w''$, where $\vert w'\vert = i$.
We consider three cases depending on the size of $p$ modulo $3$.

\medskip

As a first case, assume that $|p|=3n$. We consider three sub-cases depending on the position $i$ modulo $3$.

1.a) Assume that $i\equiv 2 \pmod{3}$. Then $p,q,r$ have $1$ as a prefix and the letter following $r$ in $\wt$ is the symbol $z_{i+9n}=1$. Hence, the word $1^{-1}pqr1$ occurs in $\wt$ in position $i+1$ and it is again a $2$-binomial cube. Indeed, thanks to Proposition~\ref{prop:cyc}, we have $1^{-1}p1\sim_2 1^{-1}q1\sim_2 1^{-1}r1$. This case is thus reduced to the case where $i\equiv 0 \pmod{3}$. 

1.b) Assume that $i\equiv 1 \pmod{3}$. Then $p,q,r$ have $0$ as a suffix and the letter preceding $p$ in $\wt$ is the symbol $z_{i-1}=0$. Hence, the word $0pqr0^{-1}$ occurs in $\wt$ in position $i-1$ and it is also a $2$-binomial cube. Thanks to Proposition~\ref{prop:cyc}, we have $0p0^{-1}\sim_2 0q0^{-1}\sim_2 0r0^{-1}$. Again this case is reduced to the case where $i\equiv 0 \pmod{3}$. 

1.c) Assume that $i\equiv 0 \pmod{3}$. In this case, we can desubstitute: there exist three words $p',q',r'$ of length $n$ such that $h(p')=p$, $h(q')=q$, $h(r')=r$ and $p'q'r'$ is a factor occurring in $w$.
We have $\Psi_2(p)=\Psi_2(q)=\Psi_2(r)$. 
By Proposition~\ref{prop:mh}, and since $M_h$ is invertible,
we have $\Psi_2(p')=\Psi_2(q')=\Psi_2(r')$, meaning that $w$ contains a $2$-binomial cube $p'q'r'$.

\medskip

As a second case, assume that $|p|=3n+1$. 
%Our aim is to show that this case never occurs. Proceed by contradiction and assume that such a $2$-binomial cube $pqr$ occurs. In this case, 
In this case, one of $p$, $q$ and $r$ occur in position $0$ modulo $3$, one in position $1$ modulo $3$, and one in position $2$ modulo $3$.
Suppose w.l.o.g. that $p$ occur in position $0$ modulo $3$, and $q$ in position $1$ modulo $3$. Then there are three factors $p'$, $q'$ and $r'$ in $w$, and $a,b\in \{0,1\}$ such that 
$p=h(p')\,0$, $q= a\, 1\, h(q')\, 0\, b$ and $r=1\, h(r')$.
By Lemma~\ref{lem:cubeone}, this is impossible.

\medskip

For the final case, assume that $|p|=3n+2$.
In this case again, one of $p$, $q$ and $r$ occur in position $0$ modulo $3$, one in position $1$ modulo $3$, and one in position $2$ modulo $3$.
Suppose w.l.o.g. that $p$ occur in position $0$ modulo $3$, and $q$ in position $1$ modulo $3$. Then there are three factors $p'$, $q'$ and $r'$ in $w$, and $a,b\in \{0,1\}$ such that 
$p=h(p')\,0\, a$, $q=1\, h(q')\, 0$ and $r=b\, 1\, h(r')$.
By Lemma~\ref{lem:cubetwo}, this is impossible.
%
%To conclude with the proof, if a $2$-binomial cube $pqr$ occurs in $\mathbf{z}$, from the study of the last two cases, {\em i.e.}, when $|p|=3n+1$ or $|p|=3n+2$, the only possible situation is when $|p|=3^j$ for some $j$. But, from the study of the first case, {\em i.e.}, when $|p|=3n$, we conclude inductively that $\mathbf{z}$ should contain a $2$-binomial cube $pqr$ with $|pqr|=3$. But, $\mathbf{z}$ does not contain any cube of length $3$.
\end{proof}

\begin{corollary}
    The infinite word $\mathbf{z}=001001011\cdots$ fixed point of $h:0\mapsto 001,1\mapsto 011$ avoids $2$-binomial cubes.
\end{corollary}

\end{document}